\documentclass[11pt]{article}
\usepackage{fullpage}
\usepackage{paralist}
\usepackage{float}
\usepackage{amsthm}
\usepackage{amsmath}
\usepackage{graphicx}
\usepackage{multirow}
\usepackage{boldline}
\usepackage{microtype}
\usepackage{amssymb}

\usepackage[colorlinks,linkcolor=blue,filecolor=blue,citecolor=blue,urlcolor=blue,pdfstartview=FitH,pagebackref]{hyperref}
\usepackage[nameinlink]{cleveref}

\newtheorem{theorem}{Theorem}
\newtheorem{lemma}{Lemma}
\newtheorem*{lemma*}{Lemma}
\newtheorem{corollary}[lemma]{Corollary}

\newtheorem{claim}[lemma]{Claim}

\newtheorem{definition}{Definition}
\newtheorem{observation}[lemma]{Observation}

\newcommand{\namedref}[2]{\hyperref[#2]{#1~\ref*{#2}}}
\newcommand{\sectionref}[1]{\namedref{Section}{#1}}
\newcommand{\appendixref}[1]{\namedref{Appendix}{#1}}

\newcommand{\theoremref}[1]{\namedref{Theorem}{#1}}

\newcommand{\figureref}[1]{\namedref{Figure}{#1}}

\newcommand{\tableref}[1]{\namedref{Table}{#1}}

\newcommand{\diam}{{\rm diam}}

\newcommand{\compKbeta}[1]{\mathop{\mbox{$[#1]_{\beta}^k$}}}

\usepackage{pdfsync}
\usepackage{authblk}

\begin{document}

\title{Covering Metric Spaces by Few Trees}
\author{Yair Bartal}\affil{Department of Computer Science, Hebrew University of Jerusalem, Israel. Email: \texttt{yair@cs.huji.ac.il}\thanks{
Supported in part by a grant from the Israeli Science Foundation (1817/17).}}
\author{Nova Fandina}\affil{Department of Computer Science, Hebrew University of Jerusalem, Israel. Email: \texttt{fandina@cs.huji.ac.il}\thanks{
Supported in part by a grant from the Israeli Science Foundation (1817/17).}}
\author{Ofer Neiman}\affil{Department of Computer Science, Ben-Gurion University of the Negev, Beer-Sheva, Israel.  Email: \texttt{neimano@cs.bgu.ac.il}\thanks{Supported in part by a grant from the Israeli Science Foundation (1817/17) and in part by BSF grant 2015813.}}
\maketitle

\begin{abstract}
A {\em tree cover} of a metric space $(X,d)$ is a collection of trees, so that every pair $x,y\in X$ has a low distortion path in one of the trees. If it has the stronger property that every point $x\in X$ has a single tree with low distortion paths to all other points, we call this a {\em Ramsey} tree cover.
Tree covers and Ramsey tree covers have been studied by \cite{BLMN03,GKR04,CGMZ05,GHR06,MN07}, and have found several important algorithmic applications, e.g. routing and distance oracles. The union of trees in a tree cover also serves as a special type of spanner, that can be decomposed into a few trees with low distortion paths contained in a single tree; Such spanners for Euclidean pointsets were presented by \cite{ADMSS95}.

In this paper we devise efficient algorithms to construct tree covers and Ramsey tree covers for general, planar and doubling metrics. We pay particular attention to the desirable case of distortion close to 1, and study what can be achieved when the number of trees is small.  In particular, our work shows a large separation between what can be achieved
by tree covers vs. Ramsey tree covers.

\end{abstract}

\section{Introduction}

The problem of approximating metric spaces by tree metrics has been a successful research thread in the past decades, and has found numerous algorithmic applications. This is mainly due to the fact that a tree has a very simple structure that can be exploited by the algorithm designer. While a single tree cannot provide a meaningful approximation, due to a lower bound of \cite{RR98} (the metric of the $n$ point cycle requires $\Omega(n)$ distortion for embedding into a tree), several other variants have been considered in the literature.
The purpose of this paper is to study the natural question whether there exists a small collection of trees (\emph{tree cover})  such that each pair is well preserved in at least one of them.
A natural stronger demand may be that for each point all of its interpoint distances to the rest of the metric are well preserved in one of the trees (\emph{Ramsey tree covers}).

Tree covers and Ramsey tree covers have been studied by \cite{GKR04,BLMN03,CGMZ05,GHR06,MN07}, and are useful ingredients in important algorithmic applications such as routing and distance oracles.

 Given a metric space $(X,d_X)$ and an edge-weighted tree $T$ with $X\subseteq V(T)$, for $x,y\in X$ let $d_T(x,y)$ denote the length of the path in $T$ from $x$ to $y$. We say $T$ is {\em dominating} if $d_T(x,y)\ge d_X(x,y)$ for all $x,y\in X$. A dominating tree $T$ has {\em distortion} $\alpha$ for a pair $x,y\in X$, if $d_T(x,y)\le\alpha\cdot d_X(x,y)$. In what follows, all trees we consider are always dominating (this can be assumed w.l.o.g.).

\begin{definition}[Tree cover]
Given a metric space $(X,d_X)$, for $\alpha \geq 1$ and an integer $k$, a {\em tree cover} with distortion $\alpha$ and size $k$, $(\alpha,k)$-tree cover in short, is a collection of $k$ dominating trees $T_1, \ldots, T_k$, such that for any  $u \neq v$ in $X$ there is a tree $T_i$ with distortion at most $\alpha$ for the pair $u,v$.

If for each $u\in X$ there is a tree $T_i$ with distortion at most $\alpha$ for each pair $u,v$ with $v\in X$, we call this a {\em Ramsey $(\alpha,k)$-tree cover}.
\end{definition}

If the metric is a shortest path metric of some graph $G$, and the trees are subgraphs, we call this a {\em spanning} tree cover.

The notion of tree covers is closely related to the well studied notion of spanners. In the context of  metric spaces, a {\em spanner} with distortion $\alpha$ for the metric $(X,d_X)$, is a graph $H$ with $X\subseteq V(H)$, so that for all $x,y\in X$, $d_X(x,y)\le d_H(x,y)\le\alpha\cdot d_X(x,y)$.   It is often desired that the spanner would be a sparse graph. Note that taking $H$ as the union of the trees in a (Ramsey) tree cover forms a sparse spanner with a special structure; that can be decomposed into a few trees, and every pair (or every point) has the distortion guarantee in one of these trees. In the context of graphs $H$ is usually required to be a subgraph of the original graph, and then the same holds for spanning tree covers. Spanners are basic graph constructions, have been intensively studied \cite{PS89,ADDJ90,C93,CDNS95,ADMSS95,EP04,BS03,TZ06,NS07} and have numerous applications in various settings, see e.g. \cite{A85,PU89b,ABCP93,C93,TZ01,EP15}.

A related well-studied concept is probabilistic embedding of a metric space into tree metrics.
This notion was introduced by Bartal \cite{bartal96}, and a sequence of works by Bartal, and Fakcharoenphol et al. \cite{B98,FRT04,B04} culminated in obtaining a tight $O(\log n)$ bound. The result of \cite{bartal96} already implies a probabilistic construction of tree covers of size $k$ with distortion $O(n^{2/k}\log n)$, for general metrics spaces. In Theorem \ref{thm:ramsey_upper_bound} we improve this by constructing deterministic Ramsey tree covers with almost optimal distortion (nearly matching the lower bound in Theorem~\ref{thm:strong_lower_bound}).

In the rest of the section we review known results on tree covers and Ramsey tree covers, and present the new results of this paper.

\subsection{Tree Covers}
In the context of Euclidean spanners, Arya et al. \cite{ADMSS95} used the so called dumbbell trees to build low distortion spanners. Rephrased in our context, they obtained a tree cover for Euclidean pointsets. More specifically, for any finite set of points in $d$-dimensional Euclidean space, and any parameter $\epsilon>0$, they devised a $(1+\epsilon)$-distortion tree cover with $O((d/\epsilon)^d\log(d/\epsilon))$ trees. We note that their trees are using Steiner points (i.e., points in Euclidean space that are not part of the input set), and it is not clear that such points can be removed from the spanner while maintaining  $(1+\epsilon$) distortion.

Chan et al. \cite{CGMZ05} presented tree covers for doubling metrics. The {\em doubling constant} of a metric $(X,d_X)$ is the minimal $\lambda$, so that every ball of radius $2r$ can be covered by $\lambda$ balls of radius $r$. The {\em doubling dimension} of $(X,d_X)$ is defined as $\log\lambda$, and a family of metrics is called {\em doubling} if every metric in it has doubling dimension $O(1)$. The result of \cite{CGMZ05} used hierarchical partitioning to construct a tree cover with distortion $O(\log^2\lambda)$ and $O(\log\lambda\cdot\log\log\lambda)$ trees.

The notion of spanning tree covers was introduced by Gupta et al. in \cite{GKR04}, who used these for MPLS routing.
They devised spanning tree covers for planar graphs: an exact (i.e. distortion 1) tree cover with $O(\sqrt{n})$ trees (more generally $O(r(n)\log n)$ trees for graphs admitting a hierarchical $r(n)$ size separators), and a spanning tree cover with distortion 3 and only $O(\log n)$ trees. They also showed the former result for planar graphs is tight, i.e., at least $\Omega(\sqrt{n})$ trees are needed for an exact tree cover.

\subsubsection{Our results} As a starting point for this study, we observe, that for general metrics, the number of trees of any tree covers with distortion $\alpha$ must be as large as $n^{1/\alpha}$. (This bound stems from the standard example of high girth graphs and extends a previous lower bound for spanning trees of \cite{GKR04}). Nearly optimal upper bounds are known even for Ramsey tree covers (see next subsection). The above lower bound also implies a lower bound of $\lambda^{1/\alpha}$ in any space with doubling constant $\lambda$.

One of our main results is a tree cover for doubling metrics. We develop a novel hierarchical clustering for such metrics, built in a bottom-up manner. We then use this new clustering to show that for any $0<\epsilon<1$, every metric with doubling constant $\lambda$ admits a tree cover with distortion $1+\epsilon$ and only $(1/\epsilon)^{O(\log\lambda)}$ trees. Since $d$-dimensional Euclidean space has doubling dimension $\Theta(d)$, the number of trees in the cover is therefore $(1/\epsilon)^{O(d)}$. Hence, this can be viewed as both a generalization and improvement of the result of \cite{ADMSS95}. Moreover, we improve their result in another aspect, since unlike \cite{ADMSS95} we do not require the use of Steiner points. In particular, for any $\epsilon>0$ our result provides a $(1+\epsilon)$-spanner with $n/\epsilon^{O(\log\lambda)}$ edges, that can be decomposed to a small number of trees, and where each pair has a $1+\epsilon$ stretch path in one of the trees. We note that the number of edges in this spanner is asymptotically optimal \cite{RS98}, and thus so is our result.

We then turn to obtaining a distortion-size tradeoff for tree covers of doubling spaces with arbitrary distortion $\alpha$.
We improve and extend the result of \cite{CGMZ05}; for any parameter $\alpha$, we use a more sophisticated construction of hierarchical  partitions, to build a tree cover with distortion $O(\alpha)$ and $O(\lambda^{1/\alpha} \cdot\log\lambda\cdot\log \alpha)$ trees (note that setting $\alpha=\log\lambda$ yields distortion $O(\log\lambda)$ with $O(\log\lambda\cdot \log \log \lambda)$ trees). We note that the trees obtained here are in fact \emph{ultrametrics}\footnote{An ultrametric $(U,d_U)$ is a metric satisfying a strong form of the triangle inequality, $\forall x,y,z$, $d_U(x,z)\le\max\{d_U(x,y),d_U(y,z)\}$. An ultrametric is both a tree metric and a Euclidean metric.}. This result provides a special type of spanner with distortion $O(\alpha)$ and $O(n\cdot\lambda^{1/\alpha}\cdot\log\lambda\cdot\log\alpha)$ edges, which improves the recent spanner construction of \cite{FN18}, whose number of edges was larger by a factor of $O(\log_\lambda n)$ (though their spanner has additionally bounded lightness). The lower bound mentioned above for doubling spaces implies that our tree cover bounds cannot be substantially improved.

For planar graphs with $n$ vertices, and more generally graphs excluding a fixed minor,
we apply the path-separators framework of \cite{T04,AG06}, and show that for any $\epsilon>0$, there exists a tree cover with distortion $1+\epsilon$ and $O((\log n)/\epsilon)^2$ trees. (Recall that for distortion 1, \cite{GKR04} used the planar separators of \cite{LT80}, but this requires $\Omega(\sqrt{n})$ trees.)
We also observe that certain hierarchical partitions of \cite{KLMN04} for planar (and fixed minor-free) graphs,
can be used to obtain a tree cover with $O(1)$ distortion and only $O(1)$ trees (the obtained trees are ultrametrics).

See \tableref{table:tree-cover} for a succinct comparison between our and previous results on tree covers.

\begin{table*}[ht]
\centering
\small
\begin{tabular}{ |l|l|l|l| }
\hline
Family & Reference & Number of trees & Distortion \\
\hlineB{2}
General metrics
& {\bf New} & $\Omega(n^{1/\alpha})$ & $\alpha$ \\
\hlineB{2}
\multirow{4}{*}{Doubling metrics} & \cite{CGMZ05} & $O(\log\lambda\cdot\log\log\lambda)$ & $O(\log^2\lambda)$ \\
\cline{2-4}
& {\bf New} & $(1/\epsilon)^{\Theta(\log\lambda)}$ & $1+\epsilon$ \\
\cline{2-4}
& {\bf New} & $O(\lambda^{1/\alpha}\cdot\log\lambda\cdot\log\alpha)$ & $O(\alpha)$ \\
\cline{2-4}
& {\bf New} & $\Omega(\lambda^{1/\alpha})$ & $\alpha$ \\
\hlineB{2}
\multirow{4}{*}{Planar metrics} & \cite{GKR04} & $\Theta(\sqrt{n})$ & 1 \\
\cline{2-4}
& \cite{GKR04} & $O(\log n)$ & 3 \\
\cline{2-4}
& {\bf New} & $O((\log n)/\epsilon)^2$ & $1+\epsilon$ \\
\cline{2-4}
& \cite{KLMN04}+{\bf New} & $O(1)$ & O(1) \\
\hlineB{2}
\end{tabular}
\caption{Results on tree covers for general, planar (our new results also hold for fixed minor-free graphs) and doubling metrics. (The upper bound for general metrics appears in \tableref{table:Ramsey-trees}.)}
\label{table:tree-cover}
\end{table*}

\subsection{Ramsey Tree Covers}\label{sec:RTC}
Given a metric $(X,d)$, the metric Ramsey problem asks for a large subset $S\subseteq X$ that embeds with a given distortion into a simple metric, such as a tree metric or Euclidean space. Following \cite{BLMN03}, \cite{MN07} gave a probabilistic construction that finds in any $n$ point metric $(X,d)$ a set $S\subseteq X$ of size at least $n^{1-1/\alpha}$ that embeds into an ultrametric with distortion $O(\alpha)$. In fact, the embedding has such distortion on all pairs in $S\times X$.  Applying this iteratively, \cite{MN07} obtained a collection of $O(\alpha\cdot n^{1/\alpha})$ trees, so that each point $x\in X$ has a "home tree" $T_x$ with distortion $O(\alpha)$ for every pair containing $x$. We call such a collection a { \em Ramsey tree cover}. Some further works aim at improving the leading constant in the distortion (\cite{Bar11,NT12,BGS16}) and finding a deterministic construction (\cite{Bar11}).
Recently, in the graph setting, \cite{ACEFN18} devised a {\em spanning} Ramsey tree cover, where the trees are subgraphs of the input graph. They obtained the same number of trees, but with slightly larger distortion $O(\alpha\cdot\log\log n)$.

We note that the number of trees in all previous works is $\alpha\cdot n^{1/\alpha}\ge\log n$ for any value of $\alpha$. It seems like a natural question to understand what can be achieved in the inverse tradeoff, where the number of trees, $k$,  is small. We remark that the lower bound via high girth graphs is rather weak, it implies that using $k$ trees the distortion must be only $\Omega(\log_kn)$, as that is the bound on the girth of a graph with $kn$ edges \cite{Boll04}.

\subsubsection{Our results.}
We focus on the regime where the number of trees is small. We first observe that a similar method as used in \cite{MN07} of iteratively extracting large Ramsey subspaces can be applied in this setting as well. Given any metric space $(X,d)$ on $n$ points, and a parameter $k\ge 1$, there exists a Ramsey tree cover of size $k$ (in particular, ultrametrics) and distortion $O(n^{1/k}\cdot\log^{1-1/k} n)$. We also note that the result of \cite{ACEFN18} can be translated to this setting: given a graph $G=(V,E)$ with $n$ vertices, we find a Ramsey {\em spanning} tree cover with $k$ spanning trees and distortion $O(n^{1/k}\cdot\log^{1-1/k} n\cdot\log\log n)$. 

Next, we investigate the tightness of this bound. We find a graph on $n$ vertices, such that any Ramsey tree cover with $k$ trees requires distortion $\Omega(n^{1/k})$, significantly improving the $\Omega(\log_kn)$ bound obtained from high girth graphs. This also implies that our upper bound is tight, up to lower order terms. 

Moreover, the graph we construct is series-parallel (in particular a planar graph) and also has $O(1)$ doubling dimension. Thus, our lower bound indicates a large separation between what can be achieved by a tree cover vs. a Ramsey tree cover; Our upper bounds give a tree cover with $O(1)$ trees and constant distortion for planar and doubling metrics (even $1+\epsilon$ distortion for the latter), as opposed to the $n^{\Omega(1)}$ distortion required with a constant number of Ramsey trees, for both planar and doubling metrics.

We also use a result of \cite{BLMN03} to show a lower bound for planar and doubling metrics in the low distortion regime: there are $n$-point planar (in fact, series-parallel) doubling metrics, such that any Ramsey tree cover with distortion $\alpha$ must contain at least $n^{\Omega(1/(\alpha\log \alpha))}$ trees.

Overall, for general, planar and doubling metrics, our results solve the question of covering metrics by Ramsey trees, up to logarithmic terms, in every regime of parameters. See \tableref{table:Ramsey-trees} for a concise description of previous and our results.

\begin{table*}[h]
\centering
\small
\begin{tabular}{ |l|l|l|l| }
\hline
Family & Reference & Number of trees & Distortion \\
\hlineB{2}
\multirow{2}{*}{General metrics} & \cite{MN07} & $O(\alpha\cdot n^{1/\alpha})$ & $O(\alpha)$ \\
\cline{2-4}
& {\bf New} & $k$ & $O(n^{1/k}\cdot\log^{1-1/k}n)$  \\
\hlineB{2}
\multirow{2}{*}{Planar \& doubling metrics} & {\bf New} & $k$ & $\Omega(n^{1/k})$  \\
\cline{2-4}
& {\bf New}& $n^{\Omega(1/(\alpha\log \alpha))}$ & $\alpha$\\
\hlineB{2}
\end{tabular}
\caption{Previous and our results on Ramsey trees for general, planar and doubling metrics.}
\label{table:Ramsey-trees}
\end{table*}

\subsection{Overview of Techniques}

{\bf Tree cover for doubling metrics.} The standard way to construct a ($1+\epsilon$)-spanner for doubling metrics is along the following lines \cite{GGN04,HM06}: Choose a hierarchical collection of $2^i$-nets (see \sectionref{sec:tree-cover-doubling} for definitions), and assign every vertex to its nearest net-point  at the level $i$ when it first leaves the net hierarchy; this creates a {\em net tree}. Then additional edges are added to other net-points within distance $\approx 2^i/\epsilon$. This spanner cannot be decomposed into a few trees as low distortion paths for the pairs use both the net tree and the additional edges.

We use a different approach for constructing a spanner, so that it can be decomposed to trees; We first partition the hierarchical net into a small number of well-separated sub-nets (so that in level $i$, distances between points in the sub-net are at least $2^i/\epsilon$). Then construct a tree for each hierarchical sub-net, by iterative clustering around the sub-net points in a bottom-up manner. In order to control the radius increase caused by the clustering of lower level sub-nets, we also take sufficiently large gaps between consecutive levels used in the same tree.\\
\noindent
{\bf Tree cover for minor-free graphs.} We apply the path separators of \cite{AG06}, asserting that graphs excluding a fixed minor have a separator consisting of $O(1)$ shortest paths (see \sectionref{sec:tree-cover-planar} for the precise definitions). Adding for each point $O(1/\epsilon)$ edges to each shortest path guarantees small distortion for all separated pairs \cite{T04,K02}. However, since we desire trees, we can allow each point to add only 1 edge (per tree) to the path separator. Using a simple randomized algorithm to choose these edge connections, we show that w.h.p. all pairs will have a low distortion tree.\\
\noindent
{\bf Hierarchical partitions.} 
We construct a collection of HST spaces (a special type of ultra-metric spaces, see Section \ref{sec:cover-vs-partition}) via a hierarchical probabilistic partitions similarly to \cite{CGMZ05}. Yet instead of using the basic probabilistic partitions (e.g. \cite{bartal96}) we use the probabilistic partitions of \cite{ABN11}, which have two main strong properties: The padding of the partition can be set as a parameter, which may also be a constant depending on distortion $\alpha$; The partitions are local, i.e. the probability of being padded is not affected by the structure of the clusters that are far enough. This allows showing that intersecting a bundle of such independent partitions achieves a good tradeoff when the number of scale levels is small. Combining these with the idea of bottom-up union of clusters similar to that of \cite{CGMZ05}, we are able to construct hierarchies with diameters of clusters decreasing by a \emph{constant factor} while letting the padding parameter depend on $\alpha$, obtaining more general and improved bounds.


\subsection{Related Work}
We note that Charikar et al. \cite{CCGGP98} studied a related question of bounding the number of trees sufficient for probabilistic embedding.  The result they obtain implies an exponentially weaker cover size than those that follow from \cite{bartal96} and from our construction.

In \cite{GHR06}, Gupta et al. considered a stronger version of tree covers (stronger than Ramsey tree covers), which they used to devise an oblivious algorithm for network design problems. We note that the lower bounds given in this paper show that the bound they get using this method is almost tight, even for doubling or planar metrics.

In the context of spanning trees, the problem of computing a spanning tree with low average stretch was first studied by Alon et al. \cite{AKPW95}.
Following Elkin et al. \cite{EEST05}, Abraham et al. \cite{ABN07,AN12} obtained a nearly tight $O(\log n\cdot\log\log n)$ bound.

\subsection{Organization}
In \sectionref{sec:ramsey-cover-upper} we show (spanning) Ramsey tree cover for general metrics with few trees.
We show tree covers with distortion $1+\epsilon$: for doubling metrics in \sectionref{sec:tree-cover-doubling}, and for planar (more generally minor-free) metrics in \sectionref{sec:tree-cover-planar}.
In \sectionref{sec:cover-vs-partition} we provide our construction of small size hierarchical family for any distortion $\alpha >2$ for doubling metrics. In addition, we describe the connection between hierarchical partitions and tree covers, and derive results on tree covers for planar (minor-free) metrics.
Our lower bounds on (Ramsey) tree covers appear in \sectionref{sec:lower}. We conclude in \sectionref{sec:open} with some open problems.

\section{Ramsey Tree Covers for General Metrics with Few Trees}\label{sec:ramsey-cover-upper}
In this section we show a deterministic  Ramsey tree cover construction for general metrics. Unlike previous works, we build a cover with a small (possible constant) number of trees.
\begin{theorem}
\label{thm:ramsey_upper_bound}
For any $n$-point metric $(X,d)$ and any $k \geq 1$, there is a deterministic algorithm that constructs a Ramsey tree cover for $X$ of size $k$ with distortion $O(n^{1/k} \cdot(\log n)^{1-1/k})$.
\end{theorem}

Our deterministic construction follows directly from the following theorem on deterministic Ramsey embedding into a tree metric that was presented in Bartal \cite{Bar11} and in Abraham et al.  \cite{ACEFN18} (alternatively, a randomized construction can be based on \cite{MN07}).

\begin{theorem}{\cite{Bar11, ACEFN18}}
\label{thm:ramsey_part_tree}
Let $(X,d)$ be a metric space, fix any subset $S\subseteq X$, and let $\alpha\geq 1$ be a parameter. There is a deterministic algorithm that finds a subset $Z \subseteq S$, of size $|Z| \geq |S|^{1-\frac{1}{\alpha}}$, and an embedding $f$ of $X$ into an ultrametric $T$ with distortion $O(\alpha)$ for any pair $(u,v) \in  Z \times X$.
\end{theorem}
\begin{proof}[Proof of \theoremref{thm:ramsey_upper_bound}]

Let $S_1=X$. For $i=1,\dots,k-1$ iteratively apply the algorithm of Theorem \ref{thm:ramsey_part_tree} on the subset $S_i$ with parameter $\alpha$ (to be determined later), and obtain trees $T_1,\dots,T_{k-1}$. The last tree $T_k$ will be constructed separately. Let $Z_i\subseteq S_i$ be the set of size at least $|S_i|^{1-1/\alpha}$ guaranteed by Theorem \ref{thm:ramsey_upper_bound}, and define $S_{i+1}=S_i\setminus Z_i$. Note that every point $x\in X\setminus S_k$ has a tree with distortion $O(\alpha)$  for all pairs in $\{x\}\times X$.
For each $i \geq 1$, we have $|S_{i+1}|=|S_i|-|Z_i| \leq |S_i|\left(1-n^{-1/\alpha}\right)$. Therefore, after $k-1$ iterations we have $|S_k| \leq n\left(1-n^{-1/\alpha}\right)^{k-1}$. We can embed the metric
$S_k$ into an ultrametric with distortion $|S_k|-1$ by the embedding of \cite{BLMN03,HM06}. This embedding can be extended to all of $X$, with distortion $O(|S_{k}|)$ for pairs in $S_k\times X$ by \cite[Lemma 4.1]{MN07}. Therefore, $\alpha$ should be chosen so that $n \cdot\left(1-n^{-1/\alpha}\right)^{k-1} \leq \alpha$. Using the inequality $e^x \geq 1+x$ for all $x\in \mathbb{R}$ we have:
$n \cdot\left(1-n^{-1/\alpha}\right)^{k-1} =n\cdot \left(1-e^{-\frac{\ln n}{\alpha}}\right)^{k-1} \leq n \cdot \left( \frac{\ln n}{\alpha}\right)^{k-1}.$
Taking $\alpha = n^{1/k} \cdot (\ln n)^{\left(1-\frac{1}{k}\right)}$ gives a Ramsey tree cover with distortion $O(\alpha)$.

\end{proof}

\section{$(1+\epsilon)$-Distortion Tree Covers for Doubling Metrics}\label{sec:tree-cover-doubling}

In this section we devise a tree cover for doubling metrics with distortion arbitrarily close to 1.
Let $(X,d)$ be a metric with doubling constant $\lambda$, and fix $0<\epsilon<1/8$.

\begin{definition}
An $r$-net $N\subseteq X$ is a set satisfying: 1) For every $x,y\in N$, $d(x,y)>r$, and 2) For every $u\in X$ there exists $x\in N$ with $d(x,u)\le r$. We say that a collection $\{N_i\}$ of $2^i$-nets is {\em hierarchical} if $N_{i+1}\subseteq N_i$.
\end{definition}
It is well-known that a simple greedy algorithm can construct (hierarchical) nets. Also, it is known that the size of an $r$-net of a ball of radius $R$ is bounded by $\lambda^{O(\log(R/r))}$ (see e.g. \cite{GKL03}).

Let $\{N_i\}$ be a hierarchical collection of $2^i$-nets of $X$. (It suffices to take the indices $i$ from the range $[\log(\epsilon\delta),\log\Delta]$ where $\delta=\min_{x\neq y\in X}\{d(x,y)\}$ and  $\Delta=\max_{x,y\in X}\{d(x,y)\}$.)
\begin{claim}
There is a partition of $N_i$ to $t=\lambda^{O(\log(1/\epsilon))}$ sets $N_{i1},\dots,N_{it}$, so that for every $x,y\in N_{ij}$, $d(x,y)\ge 6/\epsilon\cdot 2^i$. It is also hierarchical: if $x\in N_{ij}$ then $x\in N_{i'j}$ for every $i'<i$.
\end{claim}
\begin{proof}
First place in $N_{ij}$ all the points of $N_{(i+1)j}$ for each $j$, and denote $N'_i=N_i\setminus (\bigcup_j N_{(i+1)j})$. Next, for $j=1,2,...,t$ complete $N_{ij}$ by choosing greedily from the points remaining in $N'_i\setminus(N_{i1}\cup\dots\cup N_{i(j-1)})$. Since for any $x\in N_i$ the ball of radius $6/\epsilon\cdot 2^i$ contains less than $t$ net points of $N_i$, we will surely pick $x$ to some $N_{ij}$ in some iteration $j\le t$.
\end{proof}
\noindent
{\bf Construction of trees.}
Assume w.l.o.g that $\log(1/\epsilon)$ is an integer, we will construct $t\cdot\log(1/\epsilon)$ trees (in fact, forests). The tree $T_{j,p}$ is indexed by the pair $(j,p)$ with $1\le j\le t$ and $0\le p< \log(1/\epsilon)$. Fix $j$ and $p$, we now describe how to build $T_{j,p}$. Let $I_p=\{i~:~i\equiv p~({\rm mod}~\log(1/\epsilon))\}$. Initially all points in $X$ are unclustered. We go over all $i\in I_p$ (from small to large in order), and for every $x\in N_{ij}$ we add an edge from $x$ to every unclustered point $y \in X$ satisfying $d(x,y)<3/\epsilon\cdot 2^i$, of weight $d(x,y)$. These points connected to $x$ are now clustered. (The center $x$ is not considered clustered.)

\begin{observation}\label{Observ}

\begin{enumerate}
\item[{ a.}] No point $x\in N_{ij}$ is clustered when iteration $i$ is complete.
\item[b.] Every $u\in X$ can be clustered by at most 1 point.
\item[c.] If $C_x$ is the connected component created by the clustering of $x\in N_{ij}$ at level $i\in I_p$, then $\diam(C_x)\le 8/\epsilon\cdot 2^i$.
\end{enumerate}
\end{observation}
\begin{proof}

\begin{enumerate}
\item[a.] Since the nets are hierarchical, $x\in N_{i'j}$ for all $i'\le i$, $i'\in I_p$. Every other point $y\in N_{i'j}$ is at least $6/\epsilon\cdot 2^{i'}$ far away from $x$, thus $y$ will not cluster $x$.
\item [b.] We need to show that if $u$ becomes clustered at level $i$, then there exists a single $x\in N_{ij}$ who clustered it. Note that by definition $d(x,u)\le3/\epsilon\cdot 2^i$, while the separation of the sub-net $N_{ij}$ guarantees that any $x'\in N_{ij}$ has $d(x,x')>6/\epsilon\cdot 2^i$. The triangle inequality implies that $d(x',u)>3/\epsilon\cdot 2^i$, as required.

\item[c.] This follows by a simple induction, every cluster center at level $i'<i$ with $i'\in I_p$ has $i'\le i-\log(1/\epsilon)$, and so $2^{i'}\le \epsilon\cdot 2^i$. Thus by induction its diameter is at most $8\cdot 2^i$. The distance from $x$ to any point in $C_x$ is at most $3/\epsilon\cdot 2^i+8\cdot 2^i$ (the clustering range of $x$ plus the diameter of the previous level cluster). We get that $\diam(C_x)\le 2\cdot(3/\epsilon\cdot 2^i+ 8\cdot 2^i)<8/\epsilon\cdot 2^i$ whenever $\epsilon<1/8$.
\end{enumerate}
\end{proof}
\begin{claim}
When the process completes we have a forest.
\end{claim}
\begin{proof}
By Observation \ref{Observ}({$b$}) every point $u$ adds at most a single edge to $T_{j,p}$, at the time it becomes clustered. As $u$ adds this edge to an unclustered point, it cannot close a cycle. (More formally, if we give each vertex a time stamp which is the time it becomes clustered, then the single edge every vertex adds is to a vertex with a higher time stamp.)
\end{proof}

\begin{claim}\label{claim:clus}
Let $C_x$ be the connected component created when we clustered points to $x$ at some level $i\in I_p$. Then for every point $y\in B(x,2/\epsilon\cdot 2^i)$ we have $d_{C_x}(x,y)\le d(x,y)+2^{i+4}$.
\end{claim}
\begin{proof}
If $y$ was unclustered when creating $C_x$ then $d_{C_x}(x,y)= d(x,y)$ by definition. Otherwise, let $z$ be the (unique) unclustered point in $C_z$, the connected component containing $y$ before executing the clustering of level $i$. Let $i'<i$ be the level in which $z$ created $C_z$, and by Observation \ref{Observ}($c$) we have $\diam(C_z)\le 8/\epsilon\cdot 2^{i'}\le 8\cdot 2^i$ (recall $i'\le i-\log(1/\epsilon)$ since $i'\in I_p$). As $8\le 1/\epsilon$ it follows that $d(x,z)\le d(x,y)+d(y,z)\le 2/\epsilon\cdot 2^i+8\cdot 2^i\le 3/\epsilon\cdot 2^i$, so
$x$ will cluster $z$ (recall that by Observation \ref{Observ}($b$) no other center can cluster $z$). Furthermore,
$d_{C_x}(x,y)= d_{C_x}(x,z)+d_{C_z}(z,y)= d(x,z)+d_{C_z}(z,y)\le d(x,y)+2d_{C_z}(z,y)\le d(x,y)+ 2^{i+4}$.

\end{proof}
\begin{lemma}
For every $u,v\in X$, there is a tree $T=T_{j,p}$ so that $d_T(u,v)\le(1+O(\epsilon))\cdot d(u,v)$.
\end{lemma}
\begin{proof}
Choose $i$ such that $2^{i-1}/\epsilon\le d(u,v)<2^i/\epsilon$, and let $p=i~{\rm mod}~\log(1/\epsilon)$. Let $x\in N_i$ be the nearest net point to $u$, so that $d(x,u)\le 2^i$, and let $1\le j\le t$ be such that $x\in N_{ij}$.

Observe that $d(x,v)\le d(x,u)+d(u,v)\le 2^i+2^i/\epsilon<2/\epsilon\cdot 2^i$, so by Claim \ref{claim:clus} \[d_{C_x}(x,v)\le d(x,v)+ 2^{i+4}\le d(x,u)+d(u,v)+2^{i+4}=(1+O(\epsilon))\cdot d(u,v).\] We also have by the same claim that $d_{C_x}(x,u)\le d(x,u)+2^{i+4}= O(\epsilon)\cdot d(u,v)$. The fact that $C_x$ is a subtree of the forest $T=T_{j,p}$ completes the proof.
\end{proof}
Since we use edge weights that are the actual distances in $(X,d)$, clearly $d_T\ge d$. Rescaling $\epsilon$ by a constant yields the following.
\begin{theorem}
For every metric $(X,d)$ with doubling constant $\lambda$, and any $0<\epsilon<1$, there is an efficient algorithm to construct a tree cover of size $\lambda^{O(\log(1/\epsilon))}$, with distortion $1+\epsilon$.
\end{theorem}

\section{$(1+\epsilon)$-Distortion Tree Covers for Planar and Minor-Free Graphs}\label{sec:tree-cover-planar}

In this section we use {\em path-separators} for planar \cite{T04} and more generally minor-free graphs \cite{AG06}, to devise tree covers with $1+\epsilon$ distortion and $O((\log n)/\epsilon)^2$ trees. We start with some preliminary definitions.

A graph $G$ has $H$ as a minor if one can obtain $H$ from $G$ by a sequence of edge deletions, vertex deletions and edge contractions. The graph $G$ is $H$-minor-free if it does not contain $H$ as a minor.
\begin{definition}\label{def:path-separator}
A graph $G=(V,E)$ on $n$ vertices is {\em $s$-path separable} if there exists an integer $t$ and a separator $S\subseteq V$ such that:
\begin{enumerate}
\item $S=V({\cal P}_0)\cup V({\cal P}_1)\cup\dots\cup V({\cal P}_t)$, where for each $0\le i\le t$, ${\cal P}_i$ is a collection of shortest paths in the graph $G\setminus(\bigcup_{0\le j<i}{\cal P}_j)$ (and $V({\cal P}_i)$ is the vertex set used by the paths in ${\cal P}_i$).

\item $\sum_{i=0}^t|{\cal P}_i|\le s$, that is, the total number of paths is at most $s$.

\item Each connected component of $G\setminus S$ is $s$-path separable and has at
most $n/2$ vertices.
\end{enumerate}
\end{definition}
\begin{theorem}[\cite{AG06}]\label{thm:path-separator}
Every $H$-minor-free graph is $s$-path separable for some $s = s(H)$, and
an $s$-path separator can be computed in polynomial time.
\end{theorem}

The following Lemma is implicit in the works of \cite{K02,T04}, we provide a proof in \appendixref{app:portal-proof} for completeness.
\begin{lemma}\label{lem:landmarks}
Let $G=(V,E)$ be an edge-weighted graph, fix any $0<\epsilon<1$, and let $P$ be a shortest path in $G$. Then one can find for each $x\in V$ a set of landmarks $L_x$ on $P$ of size $|L_x|=O(1/\epsilon)$, such that for any $x,y\in V$ whose shortest path between them intersects $P$, there exists $u\in L_x$ and $v\in L_y$ satisfying
$d_G(x,u)+d_P(u,v)+d_G(v,y)\le(1+\epsilon)\cdot d_G(x,y)$.
\end{lemma}
\noindent
{\bf Construction.}
Using these tools, we are ready to describe our tree cover for minor-free graphs. Apply the path separator of Theorem \ref{thm:path-separator} on the input graph $G=(V,E)$, $|V|=n$, to obtain a collection ${\cal P}$ of $s$ paths, and denote $S=V({\cal P})$. For each path $P\in{\cal P}$, apply Lemma \ref{lem:landmarks} to get a set of landmarks for each vertex, and let $\ell=\max_{x\in V}\{|L_x|\}=O(1/\epsilon)$ be the maximal size of a landmark set. Let $T$ be a tree formed by taking $P$, and for each $x\in  V\setminus V(P)$ add a single edge to $u\in L_x$ chosen uniformly  and independently at random. Let $d_G(x,u)$ be the weight of a chosen edge. We pick $(C\log n)/\epsilon^2$ such trees independently for each path $P$, for sufficiently large constant $C$.

Next, we continue recursively on each connected component of $G\setminus S$. Since the number of vertices halves at every iteration, there will be $O(\log n)$ iterations. Furthermore, the trees of different connected components can be viewed as a forest of $G$ (which can be arbitrarily completed to a tree), thus the total number of trees is $O((\log n)/\epsilon)^2$.\\ \\
\noindent
{\bf Analysis.}
Fix some $x,y\in V$, and let $P$ be the first path in ${\cal P}$ that intersects the shortest path between $x,y$ in $G$. (It may be the case that ${\cal P}$ is a path separator in a deep level of the recursion, when we decompose some subgraph $G'$. Note that $d_{G'}(x,y)=d_G(x,y)$, since no path intersected the shortest path from $x$ to $y$ so far. So w.l.o.g. we call the current graph $G$.) Let $u\in L_x$ and $v\in L_y$ be such that $d_G(x,u)+d_P(u,v)+d_G(v,y)\le(1+\epsilon)\cdot d_G(x,y)$, which are guaranteed to exist by Lemma \ref{lem:landmarks}. If we choose a tree $T$ that contains $P$ and both edges $(x,u),(y,v)$, then $T$ will have distortion at most $1+\epsilon$ for the pair $x,y$. The probability that both $x,y$ add these edges to $T$ is at least $1/\ell^2=\Omega(\epsilon^2)$. Thus, the probability that none of the trees created for the path $P$ has distortion at most $1+\epsilon$ for the pair $x,y$ is at most $(1-\Omega(\epsilon^2))^{(C\log n)/\epsilon^2}\le e^{-3\ln n}=1/n^3$,
whenever $C$ is sufficiently large.
By the union bound over the ${n\choose 2}$ pairs, with high probability all pairs have a tree with distortion $1+\epsilon$ in that tree. We have proven the following.

\begin{theorem}
Let $G$ be a graph on $n$ vertices that is $H$-minor-free. For any $0<\epsilon<1$, there is a randomized efficient algorithm that w.h.p. constructs a tree cover for $G$ containing $O((\log n)/\epsilon)^2$ trees with distortion $1+\epsilon$. (The constant in the $O$-notation depends on $|H|$.)
\end{theorem}

\section{Tree Covers  for Doubling Metrics with Distortion-Size Tradeoff}\label{sec:cover-vs-partition}
In this section we prove that any metric space with doubling constant $\lambda$  has a tree cover with distortion $O(\alpha)$ of size $O(\lambda^{1/\alpha}\log \lambda \log \alpha)$, and also that graphs excluding a fixed minor have a tree cover with distortion $O(1)$ of size $O(1)$.

Recall that ultrametric is a metric space obeying a strong form of the triangle inequality. It is well known that any finite ultrametric $(U,\rho)$ can be represented by a finite labeled tree $T$, with the points of $U$ being the leaves of $T$. Each node $ u \in T$ has a label $\Delta(u) \geq 0$ and the label of each leaf is $0$. For any two nodes $u$ and $v$, such that $v$ is a child of $u$, $\Delta(u) \geq \Delta(v)$. For $u, v \in U$, the distance $\rho(u,v)$ is defined to be the label of their least common ancestor. We refer to ultrametrics by their tree representation. If the labels in an ultrametric tree $T$ are decreasing by a factor at most $\mu >1$, then $T$ is called a $\mu$-Hierarchically Separated Tree metric ($\mu$-HST) \cite{bartal96}.  We note that an ultrametric space can also be represented as a shortest path metric on a Steiner tree.

For a finite metric $(X,d)$, let $d_{\max}=\max_{x \neq y \in X}\{d(x,y)\}$, and $d_{\min}=\min_{x \neq y \in X}\{d(x,y)\}$. Let $\Phi(X):={d_{\max}}/{d_{\min}}$ denote the aspect ratio of $X$.

\subsection{Probabilistic Hierarchical Partition Family}
We start with the necessary definitions. For any $\Delta >0$, a $\Delta$-bounded partition ${P}$ of a finite metric space $(X,d)$ is a collection of pairwise disjoint clusters $P_i \subseteq X$, such that $\cup P_i=X$, and for each cluster $P_i \in P$, $\diam(P_i) \leq \Delta$.  We assume that each cluster has some point designated as its center. For a point $x\in X$, let $P(x) \in P$ denote the cluster that contains $x$. A $\Delta$-bounded probabilistic partition of $X$ is a distribution $\mathcal{P}$ over a set of $\Delta$-bounded partitions of $X$.

The notion of a padding parameter of a random partition is studied in various papers \cite{LS91,KPR93,bartal96,FRT04}. We use a stronger definition given by Abraham et al. in \cite{ABN11}, where the padding parameter depends on the desired probability of success. The following is a rephrased version of their original definition (\cite{ABN11}, Definition $17$):
\begin{definition}[Padded Probabilistic Partition]
Let $\eta(\delta)\colon (0,1] \to (0,1]$ be some function, and $(a, b] \subseteq (0, 1]$ be some range. A $\Delta$-bounded probabilistic partition $\mathcal{P}$ is $\eta(\delta)$-padded on the range $(a,b]$, if for all $x\in X$ and for all $\delta\in (a, b]$, $\Pr\limits_{{P \sim \mathcal{P}}}\left[ B(x,\eta(\delta)\cdot\Delta) \subseteq P(x) \right] \geq \delta$.
\end{definition}

In addition, the authors defined a notion of a {\it locally} padded probabilistic partition (on the range $(a, b]$): $\mathcal{P}$ is $\eta(\delta)$-locally padded if for all $a <\delta \leq b$ the event $B(x,\eta(\delta)\cdot\Delta) \subseteq P(x)$ occurs with probability at least $\delta$ regardless of the structure of the partition outside the ball $B(x, 2\Delta)$. Formally stated, for all $x\in X$, for all subsets $C \subseteq X\setminus B(x,2\Delta)$ and all partitions $P'$ of $C$,
$\Pr\limits_{P \sim \mathcal{P}}\left[B(x, \eta(\delta)\cdot \Delta) \subseteq P(x)\;|\;P[C]=P' \right] \geq \delta$,
where $P[C]$ denotes the restriction of the partition $P$ to $C$.  Our construction uses their random partitions as a building block:

\begin{lemma}[\cite{ABN11}, Lemma $8$]\label{claim:local_one_scale}
Given a finite metric space $X$ with doubling constant $\lambda$, and given any $0<\Delta< \diam(X)$, there is a $\Delta$-bounded, $\left(\frac{\log (1/\delta)}{2^6 \log \lambda}\right)$-locally padded probabilistic partition $\mathcal{P}$ of $X$, for $\delta \in \left[\lambda^{{-2}^{12}},1 \right]$.
\end{lemma}

A set of nested partitions of $X$ forms a hierarchy:

\begin{definition}[Hierarchical Partition]\label{def:HP}
For all $\mu>1$, $\Delta \leq d_{\max}(X)$ and integer $ 1\leq B \leq \log_{\mu}\Phi(X)$,  let $\Delta_i=\Delta/\mu^i$, for all $0 \leq i \leq B$. A $\mu$-Hierarchical Partition of $X$ for range $[ \Delta, \Delta_B]$, is a collection $H = \{P_0, \ldots, P_B\}$ of  partitions of $X$ such that: For all $0 \leq i\leq B$, $P_i$ is a $\Delta_i$-bounded partition of $X$; Each $P_{i+1}$ is a refinement of $P_{i}$, i.e. each cluster in $P_{i}$ is a union of some clusters in $P_{i+1}$.  Let $\mu$-${HP}_B(\Delta)$ denote such a collection.

A full range Hierarchical Partition, denoted by $\mu$-HP, is the $\mu$-${HP}_B(\Delta)$, for $\Delta=d_{\max}(X)$ and $B=\log_{\mu}\Phi(X)$ (we assume this is an integer).
\end{definition}
There is a natural way to associate a dominating $\mu$-HST tree to a $\mu$-HP. For each cluster of the partition $P_i$ there is a node in the tree. The nodes associated with clusters of the partition $P_{i+1}$ are the children of nodes associated with clusters of $P_{i}$. The label of all level $i$ nodes in the tree is $ \Delta/\mu^i$. The points of $X$ are at the leaves.


\begin{definition}[$\eta$-Padded $\mu$-Hierarchical Partition Family ]
Let $\eta <1$ and $\mu>1$. For a finite metric space $(X,d)$, an $\eta$-padded $\mu$-Hierarchical Partition Family of $X$, $(\eta, \mu)$-HPF, is a set $\mathcal{H}$ of $\mu$-Hierarchical Partitions $\{H^j\}_{j\geq 1}$ of $X$ such that: For all $x \in X$ and for all scales $0\leq i \leq \log_{\mu}\Phi(X)$, there is an $H^j \in \mathcal{H}$ such that $B(x, \eta\Delta_i) \subseteq P_i^{(j)}(x)$, where $P_i^{(j)}$ is a $\Delta_i$-bounded partition of $H^j$. The size of $\mathcal{H}$ is the number of hierarchical partitions it has.

\end{definition}
The following lemma shows the connection between hierarchical family and a tree cover:



\begin{lemma}\label{lemma:cover_from_HP}
If there is an $(\eta, \mu)$-HPF of size $k$ of $X$, then there is an $({\mu}/{\eta}, k)$-tree cover of $X$.
\end{lemma}
\begin{proof}
Let $H^{1}, \ldots, H^{k}$ be an $(\eta, \mu)$-HPF of $X$. Consider an associated collection of dominating $\mu$-HST trees $T_1, \ldots, T_k$.  Given any $x\neq y \in X$, let $i$ be the minimal index such that $d(x,y) \geq \eta\Delta_i$. If $i=0$, then by the construction, for any tree $T_j$, $d_{T_j}(x,y) \leq \Delta_0$, implying $d_{T_j}(x,y)/d(x,y)\leq 1/\eta$.   If $i \geq 1$, then $ \eta \Delta_i\leq d(x,y) \leq \eta \Delta_{i-1}$. The padding property implies that  there is $H^{j}$ such that $B(x, \eta \Delta_{i-1}) \subseteq P_{i-1}^{(j)}(x)$.  As  $y\in B(x, \eta \Delta_{i-1})$, it holds that $d_{T^j}(x,y) \leq \Delta_{i-1}$. Therefore, ${d_{T^j}(x,y)}/{d(x,y)} \leq {\Delta_{i-1}}/{\eta \Delta_{i}} = {\mu}/{\eta}$.
\end{proof}
In what follows, we construct $(\Omega(1/\alpha), 2)$-HPF of $X$, of size $O(\lambda^{1/\alpha}\log \lambda \log \alpha)$.
We note that the notion of hierarchical family also appeared in \cite{KLMN04}, where the authors constructed an $(\Omega(s^{-2})),O(s^2))$-HPF of size $3^s$ for any metric of a $K_{s,s}$-minor free graph. As a corollary, we conclude
\begin{corollary}
For any metric induced on a $K_{s,s}$-minor free graph, there is a tree cover with distortion $O(s^4)$, of size $3^s$.
\end{corollary}

In our proofs we will use the following version of the Lovasz Local Lemma:
\begin{lemma}[\cite{EL75}]\label{lemma:lll}
Let $\mathcal{E}_1, \ldots, \mathcal{ E}_n$ be a family of events. Let $G(V,E)$ be a directed graph on $n$ vertices with out-degree at most $d$, where each vertex corresponds to an event. Assume that for all $1\leq i\leq n$, for all $Q \subseteq \{j\;|\; (\mathcal{E}_i, \mathcal{E}_j) \notin E\}$, $\Pr[\mathcal{E}_i\;|\;\bigcap_{j \in Q}{\neg{\mathcal{E}}_j}] \leq p$. If $ep(d+1)\leq 1$, then $\Pr\left[ \bigcap_{i\in [1,n]} \neg{\mathcal{E}}_i \right]>0$.
\end{lemma}

\subsection{Constructing Hierarchical Padded Family of Bounded Size}\label{subsec:const:family}
Our main hierarchical partitions result is:

\begin{theorem}\label{thm:main:doubl:larg:dits}
For any finite metric space $X$ with doubling constant $\lambda$ and for any $\alpha \geq 2$, there is an $\Omega(1/\alpha)$-padded $2$-Hierarchical Partition Family of $X$, of size $O\left(\lambda^{1/\alpha}\log \alpha \log \lambda \right)$.
\end{theorem}

Note that taking $\alpha=O(\log \lambda)$, we obtain a hierarchical family with padding $\Omega(1/\log \lambda)$, of size $O(\log \lambda \log \log \lambda)$, which is an improvement over the result of \cite{CGMZ05}: $O(\log \lambda)$-hierarchical partitions with padding $\Omega(1/\log\lambda)$, of the same size. They construct a family of hierarchies, where each hierarchy is constructed in a bottom-up manner: the clusters of larger diameters are the union of the clusters of lower diameters. Preserving the padding parameter requires the diameters of the clusters to increase by a factor of $O(\log \lambda)$, thus covering only $\log \log \lambda$ of all the distance scales in the metric space. This results in $O(\log^2 \lambda)$ distortion. Using the Lovasz Local Lemma they were able to bound the size of this family.

In our construction, we combine the bottom-up union of clusters technique with an {\it intersection of clusters} procedure. Essentially, there are two steps. First, we use the locally padded partitions of Lemma \ref{claim:local_one_scale} to create a padded hierarchy with diameters decreasing by a constant factor, by intersecting the clusters of  levels of the hierarchy of larger diameter. Using the locality property and the fact that the padding parameter depends on the success probability, we show that using  $\log \log \lambda$ such levels of partitions with diameters increasing by factor $2$, results in a $2$-hierarchy with $\Omega(1/\alpha)$ padding, thus covering the $\log\log \lambda$ scales uncovered by the construction of \cite{CGMZ05}. We apply the Lovasz Local Lemma to bound the size of the family of such hierarchies by $O\left(\lambda^{1/\alpha}\log \alpha \log \lambda \right)$. Second, we combine the hierarchies obtained by cutting clusters, in a bottom-up manner, by defining higher scales clusters as the union of lower level clusters,  thus obtaining a hierarchy with diameters decreasing by a factor of 2  in all its levels, while padding is $\Omega(1/\alpha)$.

To prove Theorem \ref{thm:main:doubl:larg:dits}, we consider hierarchical partitions that cover a range of scales:
for any $\Delta$ and an integer $B$ we build a family $\{H^{j}\}_{j\geq 1}$, where each $H^j$ is a $\mu$-HP$_{B}$$(\Delta)$. The padding property is then required to hold for all points $x\in X$ and for all scales $\Delta_i \in [\Delta, \Delta_B]$. We call such family as $(\eta, \mu)$-HPF for range $[\Delta, \Delta_B]$.

The following lemma is used as a subroutine in the construction of the hierarchical family:
\begin{lemma}\label{claim:block:B}
Let $X$ be a finite metric space with doubling constant $\lambda$ . For a given $\alpha \geq 2$, $\Delta \leq \diam(X)$ and an integer $ 1\leq B \leq \log_{\mu}\Phi(x)$, there exists an $(\Omega(1/\alpha), 2)$-HPF for range $[\Delta, \Delta_B]$, of size $O\left( \lambda^{1/\alpha}\log \lambda (\log \alpha + B) \right)$.
\end{lemma}
\begin{proof}
For a given distortion $\alpha \geq 2$, let $\delta = \lambda^{-1/(2\alpha)}$. Therefore, for such $\delta$ we have $\eta(\delta):=\frac{\log (1/\delta)}{2^6\log \lambda}=2^{-7}/\alpha$. Also note that for any $\alpha \geq 1$, it holds that $\delta\in[\lambda^{-2^{12}},1]$. Thus, we will show that there exists a hierarchical family with padding $\Omega(\eta(\delta))$, of size $k:=O(\left( \lambda^{1/\alpha}\log \lambda (\log \alpha + B) \right))$.

Let $N \subseteq X$ be an $(\eta(\delta)\Delta_B/4)$-net of $X$. We show the claim is true for $N$ and the extension of it to $X$ is immediate, with a constant factor loss in distortion. In the sequel, all the balls are balls of metric space $N$.

Let $\Delta_i=\Delta/2^i$, for all $0\leq i \leq B$.
Consider the following random process: For each scale $\Delta_i$ in the range $[\Delta, \Delta_B]$, independently generate $\Delta_i$-bounded partitions $P_0, \ldots P_B$ of $N$ by invoking the locally padded probabilistic decomposition of Lemma \ref{claim:local_one_scale}. To obtain a $2$-Hierarchical Partition $H$ for the scales $[\Delta_0, \Delta_B]$ we cut all the clusters of all the partitions, to get $\Delta_i$-bounded nested partitions $\hat{P}_0, \ldots, \hat{P}_B$. Let $\hat{P}_0=P_0$, for all $i \geq 1$, define
$\hat{P}_i= \cup_{\hat{C} \in \hat{P}_{i-1}} \cup_{C \in P_i} C \cap \hat{C}$.

Now, independently repeat the above random process $k$ times to obtain a randomly generated family $H^{(1)}, \ldots, H^{(k)}$ of $2$-Hierarchical Partitions of the net $N$, for range $[\Delta, \Delta_B]$. Each hierarchical partition $H^{(t)}$ consists of $\Delta_i$-bounded partitions, denoted by $\hat{P}_i^{(t)}$.

For each $x \in N$ and for each scale $\Delta_i \in [\Delta, \Delta_B]$, let $\mathcal{E}_{x,i}$ be an event that the ball $B(x, \eta(\delta)\Delta_i)$ is not padded at the $i$-th level partition $\hat{P}_i^{(t)}$ in any of the hierarchical partitions $H^{(1)}, \ldots, H^{(k)}$.
We use the Lovasz Local Lemma (Lemma \ref{lemma:lll}) to prove that for the chosen value of $k$, $\Pr\left[{\bigcap}_{\substack{x \in X,\\ 0\leq i \leq B}} \neg{ \mathcal{E}_{x,i}}\right]>0$.
Let $G=(V, E)$ be a directed graph with $V=\{\mathcal{E}_{x,i}\}$, for all $x\in N$ and $ 0\leq j \leq B$. The vertex $\mathcal{E}_{x,i}$ is connected with an out-edge with all the verticies $\mathcal{E}_{y,j}$, such that $y \in B(x, 2\Delta)$ and $ 0\leq j \leq B$. We prove the following lemma:
\begin{lemma}\label{appendix:lemma:hp}
For all $Q \subseteq N\setminus B(x,2\Delta)$, for all $J \subseteq [0, B]$,
$Pr[ \mathcal{E}_{x,i} \;|\;\bigcap_{\substack{y \in Q, \\ j \in J }} \neg \mathcal{E}_{y,j}] \leq (1-\delta^2)^k$. \end{lemma}

\begin{proof}
Consider a random $H^{(t)}$ in the family we have generated, let $P_0^{(t)}, \ldots, P_B^{(t)}$ be the random $\Delta_i$-bounded partitions that were generated to build $H^{(t)}$, i.e. these are partitions before we cut them to get $H^{(t)}$. Let $\hat{P}_0^{(t)}, \ldots, \hat{P}_B^{(t)}$ be the resulting partitions after the cut of clusters procedure. For each $x \in N$ and for each $\Delta_i$, define the following events:
$\hat{A}_{x,i}^{(t)}:=B(x, \eta(\delta)\Delta_i) \subseteq \hat{P}_i^{(t)}$; for $l \leq i$, $A_{x,l,i}^{(t)}:=B(x, \eta(\delta))\Delta_i \subseteq P_l^{(t)}$.

Note that since $H^{1}, \ldots, H^{B}$ were generated independently, and since for each $t$, for all $Q \subseteq N\setminus B(x,2\Delta)$ and for all $J \subseteq [0, B]$:
$\{\neg\mathcal{E}_{y,j}\;|\;y\in Q, j\in J \} \subseteq \{\hat{A}_{y,j}^{(t)}\;|\;y\in Q, j\in J \}$,
 it is enough to prove that
\[\Pr[\hat{A}_{x,i}^{(t)}\;|\;{\bigcap}_{\substack{y \in Q,\\ j \in J}} \hat{A}_{y,j}^{(t)}] \geq \delta^2.\]
Thus,
\[\Pr[\hat{A}_{x,i}^{(t)}\;|\;{\bigcap}_{\substack{y \in Q,\\ j \in J}} \hat{A}_{y,j}^{(t)}]=\Pr[ \bigcap_{0\leq l \leq i}A_{x,l,i}^{(t)}\;|\;\bigcap_{\substack{y \in Q, \\ j \in J}} \bigcap_{l' \leq j} A_{y,l', j}^{(t)} ]= \Pr[ \bigcap_{0\leq l \leq i}A_{x,l,i}^{(t)}\;|\; \bigcap_{\substack{y \in Q, \\  0\leq l' \leq B}} \bigcap_{j \in J: j\leq l'} A_{y, l', j}^{(t)}]\]

\[\geq \Pr[ \bigcap_{0\leq l \leq i}A_{x,l,i}^{(t)}\;|\; \bigcap_{\substack{y \in Q, \\  0\leq l' \leq B}} A_{y, l', l'}^{(t)}] = \prod_{0 \leq l \leq i}\Pr[A_{x,l,i}^{(t)}\;|\;\bigcap_{\substack{y \in Q, \\  0\leq l' \leq B}} A_{y, l', l'}^{(t)}],\]
where the last equality holds due to independence of $A_{x,l,i}^{(t)}$ for different $l$. Note that for all $Q \subseteq N\setminus B(x, 2\Delta)$ it must be $Q \subseteq N\setminus B(x, 2\Delta_l)$.  Therefore:
 \[\prod_{0 \leq l \leq i}\Pr[A_{x,l,i}^{(t)}\;|\;\bigcap_{\substack{y \in Q, \\  0\leq l' \leq B}} A_{y, l', l'}^{(t)}] \geq \prod_{0 \leq l \leq i}\Pr[ A_{x,l,i}^{(t)}\;|\;\bigcap_{y \in Q}A_{y,l,l}^{t} ].\]
For all $ 0\leq l \leq i$, let $\delta_l^{(i)}$ be such that $\eta(\delta_{l}^{(i)})\Delta_l=\eta(\delta)\Delta_i$. Then, 
$A_{x,l,i}^{t}:=\{B(x, \eta(\delta)\Delta_i) \subseteq P_l^{(t)}(x)\}=\{B(x, \eta(\delta_l^{i})\Delta_l) \subseteq P_l^{t}(x)\}$. The locality property of the partitions of \cite{ABN11} states that for all $Q \subseteq N\setminus B(x, 2\Delta_l)$,
$\Pr[A_{x,l,i}^{(t)}\;|\; \bigcap_{y \in Q}A_{y,l,l}^{(t)}] \geq \delta$, therefore,
\[\prod_{0 \leq l \leq i}\Pr\left[ A_{x,l,i}^{(t)}\;|\;\bigcap_{s}A_{y,l,l}^{t} \right] \geq \prod_{ 0\leq l \leq i}\delta_l^{(i)}.\]
Recall that $\eta(\delta_{l}^{(i)})\Delta_l=\eta(\delta)\Delta_i$, i.e. $\delta_{l}^{i}=\delta^{1/2^{i-l}}$. Therefore,
$\prod_{0\leq l \leq i} \delta_l^{(i)} \geq \delta^2$,
which completes the proof of the lemma.

\end{proof}

Thus, for $\delta=\lambda^{-1/(2\alpha)}$, and $k=O\left( \lambda^{1/\alpha}\log \lambda (\log \alpha + B) \right)$, the above lemma implies 
$(1-\delta^2)^k \leq e^{-\delta^2k}\leq \lambda^{-\Theta(\log \alpha + B)}$.
In addition, the out degree $d$ of $G$ is bounded by
\[d= B \cdot |N \cap B(x, 2\Delta)| = B \cdot O\left(\left(\frac{\Delta}{\eta(\delta) \Delta_B} \right)^{\log \lambda}\right) = B \cdot \lambda^{O(\log(1/\eta(\delta)) +B )} = \lambda^{O(\log(1/\eta(\delta)) +B )}.\]
Thus, the LLL can be applied to conclude the proof of Lemma \ref{claim:block:B}.
\end{proof}

\begin{proof}[Proof of Theorem \ref{thm:main:doubl:larg:dits}]
Let $\Phi=\Phi(X)$, $\Delta_0=d_{\max}(X)$, and for all $ 1\leq i \leq \log \Phi$, $\Delta_i=\Delta_0/2^i$. Let $\mathcal{I}=\{\Delta_i\;| 0\leq \;i \leq \log \Phi\}$. We build a small family of $2$-HP's, such that the padding property is satisfied for all points in $X$ and for all scales $\Delta_i \in \mathcal{I}$, with padding parameter $\Omega(1/\alpha)$. Let $B= \lceil\log (2\alpha/c') \rceil$, where $c' < 1$ will be defined later,  and let $k=O\left(\lambda^{1/\alpha}\log \alpha \log \lambda \right)$.
We will build a family $\mathcal{H} \cup\mathcal{R}$ such that:
\begin{inparaenum}
\item $\mathcal{H}$ is a collection of size $k$ of $2$-HP's. For\footnote{For the simplicity of representation, we assume that $B$ is integer and that $\log \Phi$ is a multiple of $B$.}
$\mathcal{I}_{\mathcal{H}}:= \left\{\Delta_j\;|\; \Delta_j \in \bigcup_{ 0\leq l \leq L } \left[ \Delta_{2lB}, \Delta_{(2l+1)B}\right] \right\} \subseteq \mathcal{I}$,
where $L=\log ({\Phi^{1/(2B)}})-1/2$, the following padding property holds:
 For all $x \in X$ and for all scale $\Delta_j \in \mathcal{I}_{\mathcal{H}} $,  there is a $2$-HP $H \in \mathcal{H}$ such that $B(x, \Omega(1/\alpha)\Delta_j) \subseteq P_j(x)$, for the $j$-th level partition $P_j \in H$. \item $\mathcal{R}$ is a collection of size $k$ of $2$-HP's. The padding property as for $\mathcal{H}$ holds for all $x \in X$ and for all scales $\Delta_j \in \mathcal{I}_{\mathcal{R}}:=\mathcal{I}\setminus \mathcal{I}_{\mathcal{H}}$.
\end{inparaenum}

Namely, the hierarchical partitions of $\mathcal{R}$ are padded for the scales that are not padded in the partitions of $\mathcal{H}$. Thus, together these two collections constitute an $\Omega(1/\alpha)$-padded $2$-HP Family for $X$, of size $2k$.
We describe the construction of $\mathcal{H}$, while $\mathcal{R}$ is constructed  similarly.

Let $\mathcal{H}=H^{(1)}, \ldots, H^{(k)}$ denote the set of $2$-HP's. We construct it iteratively in a bottom up fashion. Assume by induction that we have already constructed a family $\hat{\mathcal{{H}}}= \hat{H}^{(1)}, \ldots, \hat{H}^{(k)}$ such that: Each $\hat{H}^{(t)}$ is a $2$-HP for range $[\Delta_{2B}, \Delta_{(2L+1)B}]$;  The padding property holds with parameter $\Omega(1/\alpha)$ for all scales $\Delta_j \in \mathcal{I}_{\mathcal{H}}\setminus \{ \Delta_i \in [\Delta_0, \Delta_B]\}$.

Let $c=1+\frac{1}{2^{B-1}-1}$,  by Lemma \ref{claim:block:B} there is $(\Omega(1/\alpha),2)$-HPF ${F}^{(1)}, \ldots, {F}^{(k)}$ for range $[\tilde{\Delta}_0, \tilde{\Delta}_B]$, where $\tilde{\Delta}_j=\Delta_j/c$, for $ 0\leq j \leq B$.   For each $ 1\leq t \leq k$, ${H}^{(t)}$ is obtained by adding the partitions of ${F}^{(t)}$ to $\hat{H}^{(t)}$ in the following way. Let $\hat{P}_{2B}^{(t)}$ denote the $\Delta_{2B}$-bounded partition of $\hat{H}^{(t)}$. First, for all scale $\Delta_j \in [\Delta_{B+1}, \Delta_{2B}]$ add to $H^{(t)}$, $\Delta_j$-bounded partition $P_j^{(t)}:=\hat{P}_{2B}^{(t)}$ (these artificial partitions are added to have a well defined $2$-HP family). Next, let $\{\tilde{P}_j^{(t)}\}_{ 0 \leq j \leq B}$  denote the set of $\tilde{\Delta}_j$-bounded partitions of $F^{(t)}$. For all $j$ starting from $j=B$ down to $j=0$, the partition $P_j^{(t)}$ is constructed a s follows: for each $\tilde{C} \in \tilde{P}_j^{(t)}$, add a cluster $C$ to $P_j^{(t)}$,   defined by $C=\cup \{C' \in P_{2B}^{(t)}\;|\; \mbox{center of}\; C' \in \tilde{C}\}$.
Finally, the partitions of $\hat{H}^{(t)}$ are unchanged.

For all $ B\leq j \leq 2B$, $P_j{(t)}$ is $\Delta_j$-bounded, since $\Delta_{2B} \leq \Delta_j$.
For $ 0\leq j \leq B$ the diameter of each cluster in partition $P_j^{(t)}$ is bounded by $\tilde{\Delta}_j + 2\Delta_{2B} = \Delta_j/c + \Delta_{B}/2^{B-1} \leq \Delta_j\left(  1/c + 1/2^{B-1}\right)=\Delta_j$, for a chosen value of $c$. In addition, by the construction, the partitions $P_j^{(t)}$ form a hierarchy.
It is left to show that the padding property holds in $\mathcal{H}$ for the scales $[\Delta_0, \Delta_{B}]$. By Lemma \ref{claim:block:B}, for any $x \in X$, for any $\tilde{\Delta}_j \in [\tilde{\Delta}_0, \tilde{\Delta}_{B}]$, there is $F^{(t)}$ such that $B(x, (c'/\alpha)\tilde{\Delta}_j) \subseteq \tilde{P}_{j}^{(t)}(x)$, for $\tilde{P}_j^{(t)} \in F^{(t)}$, for some constant $c'$.  Consider some cluster $P_{j}^{(t)}(x)$, for some $x \in X$. In the process of constructing $P_{j}^{(t)}(x)$ some points from $\tilde{P}_j^{(t)}(x)$ may be removed, due to removal of some cluster $C' \in P_{2B}^{(t)}$ whose center falls outside the cluster $\tilde{P}_j^{(t)}(x)$.
For $B$ as defined above, for $r = \frac{c'}{\alpha}\tilde{\Delta}_j - {\Delta_{2B}} \geq \left(\frac{c'}{2\alpha}\right)\Delta_j$, we have that $B(x, r) \subseteq {P}_{j}^{(t)}(x)$. This completes the proof.

\end{proof}

\begin{theorem}\label{thm:main:dits:doubl}
For any finite metric space $(X,d)$ with doubling constant $\lambda$, for any $\alpha \geq 2$, there is a tree cover of $X$ with distortion $O(\alpha)$ and of size $O\left(\lambda^{1/\alpha}\log \alpha \log \lambda \right)$.
\end{theorem}

\begin{proof}
Apply  Lemma~\ref{lemma:cover_from_HP} on the Hierarchical Family of Theorem \ref{thm:main:doubl:larg:dits}.
\end{proof}
Note that the tree cover of Theorem \ref{thm:main:dits:doubl} can be deterministically constructed in polynomial time via the constructive local lemma due to \cite{MT09}.

\section{Lower Bounds }\label{sec:lower}

\subsection{Ramsey Tree Covers for Planar and Doubling Metrics }
In this section we use a lemma of \cite{BLMN03} to show a nearly matching lower bound for Ramsey tree covers of planar and doubling metrics in the low distortion regime.
\begin{theorem}\label{thm:lower_bound_small_new} For any parameters $n,\alpha> 2$, there exists a planar $n$-point metric space with constant doubling dimension, such that any Ramsey tree cover of it with distortion $\alpha$ must have at least $k=n^{\Omega\left(\frac{1}{\alpha \log \alpha} \right)}$ trees.
\end{theorem}

The following claim states that any lower bound for the standard metric Ramsey problem implies lower bound for Ramsey tree cover constructions:
\begin{claim}\label{claim:lower_bound_red}
Let $Y$ be an $n$ point metric space such that the largest subspace of $Y$ which embeds into a tree metric with distortion $\alpha >1$ has size at most $g(n, \alpha)$. Then, any Ramsey tree cover of $Y$ with distortion $\alpha$ must have $k \geq \frac{n}{g(n,\alpha)}$ trees.
\end{claim}

\begin{proof}
Let $T_1, \ldots, T_k$ be a Ramsey tree cover of $Y$ with distortion $\alpha$. Namely, for each point $u \in Y$ there is at least one $T_i$ that has distortion $\alpha$ for all pairs $u,v$, for all $v\in Y$; we will say that $T_i$ hosts $u$.  Consider, for some specific $i$, the set of all points hosted by $T_i$. There exists some $i$ such that the size of this set is at least $n/k$. As this set forms a subspace of $Y$ which embeds into a tree, its size is bounded by $g(n,\alpha)$, implying $n/k \leq g(n, \alpha)$.
\end{proof}

\begin{proof}[Proof of Theorem \ref{thm:lower_bound_small_new}]
Let $N$ be the minimal integer such that $N/3-1>\alpha$, and let $C_N$ be the $N$-cycle.  In \cite{BLMN03}, among other results, the authors have shown that taking the metric space $Y=[C_N]^t$, any embedding of $S\subseteq Y$ into a tree with distortion $N/3-1$ must have $|S|\le n^{\delta}$ for $\delta=\frac{\ln (N-1)}{\ln N}$, where $t$ is chosen so that $|Y|=N^t=n$. Using the inequality $\ln (x-1) \leq \ln x -1/x$, we obtain that $|S|\le n^{\frac{\ln N-1/N}{\ln N}}=n^{1-1/(N\ln N)}=n^{1-\Omega\left( \frac{1}{\alpha \log \alpha}\right)}$. Therefore, by  Claim \ref{claim:lower_bound_red} the number of trees in any Ramsey tree cover of $Y$ with distortion $\alpha$ is at least $n^{\Omega\left( \frac{1}{\alpha \log \alpha}\right)}$.
As we noted after the proof of \theoremref{thm:strong_lower_bound}, there exists a planar doubling metric which embeds to $Y$ with constant distortion, hence the lower bound holds for such metrics as well.
\end{proof}

\subsection{Lower Bound on Tree Cover for Doubling and General Metrics}

Here we show a lower bound on tree covers for doubling metrics (and general metrics), indicating that our bound in \theoremref{thm:main:dits:doubl} is nearly tight. To this end, we use the result of \cite{ADDJS93} that there are metrics on $\lambda$ points such that any $\alpha$-spanner requires at least $\Omega(\lambda^{1+1/\alpha})$ edges (and the fact that any metric with $\lambda$ points has doubling constant at most $\lambda$).

\begin{theorem}
For any parameters $n,\lambda,\alpha\ge 1$, there is an $n$-point metric space with doubling constant $\lambda$, such that any $\alpha$-spanner requires at least $\Omega(n\cdot\lambda^{1/\alpha})$ edges.
\end{theorem}
\begin{proof}
Consider the metric $(Z,d_Z)$ on $\lambda$ points such that any $\alpha$-spanner requires at least $\Omega(\lambda^{1+1/\alpha})$ edges. Define $(X,d)$ as the metric created by $n/\lambda$ copies of $Z$ arranged in a line, i.e., the distance between points in copy $i$ to points in copy $j$ is $2|i-j|\cdot\diam(Z)$, while distances between points in the same copy are those of $d_Z$. (More formally, it is the metric composition of the line metric $P_{n/\lambda}$ with $Z$, as defined in Section \ref{sec:lower}). It can be easily checked that $(X,d)$ has doubling constant at most $\lambda$, and any $\alpha$-spanner requires $\Omega(\lambda^{1+1/\alpha})$ edges per copy of $Z$, to a total of $\Omega(n\cdot\lambda^{1/\alpha})$ edges.
\end{proof}
Since a collection of trees in a tree cover is a spanner, and each tree contains less than $n$ edges, we get the following.
\begin{corollary}\label{corr:lower:bound}
There is an $n$-point metric $(X,d)$ with doubling constant $\lambda$, such that any tree cover for $X$ with distortion $\alpha$ requires at least $\Omega(\lambda^{1/\alpha})$ trees.

In particular there is an $n$-point metric such that any tree cover for $X$ with distortion $\alpha$ requires at least $\Omega(n^{1/\alpha})$ trees.
\end{corollary}

\subsection{Lower Bound on Ramsey Tree Cover}

In this section we show an asymptotically tight lower bound on Ramsey tree covers, in the regime where the number of trees is small. The lower bound on high girth graphs mentioned above can only give distortion $\Omega(\log_k n)$ for tree covers with $k$ trees. Here we use a different example (which is a planar metric with $O(1)$ doubling dimension) that strengthen the lower bound on the distortion to $\Omega(n^{1/k})$.

We will need the following notion of a composition of metric spaces that was introduced in \cite{BLMN03} (we present here a simplification of the original definition).

\begin{definition}
Let $(S,d_S)$, $(T, d_T)$ be finite metric spaces. For $\beta \geq 1/2$, the $\beta$-composition of $S$ with $T$, denoted by $Z=S_{\beta}[T]$, is a metric space of size $|Z|=|S|\cdot |T|$ constructed by replacing each point $u \in S$  with a copy of $T$, denoted by $T^{(u)}$. Let $\gamma=\frac{\max_{t \neq t' \in T} \{d_T(t, t') \}}{ \min_{s \neq s' \in S}\{ d_S(s, s')\}}$. For $z_i \neq z_j \in Z$ such that $z_i \in T^{(u)}$  and $z_j\in T^{(v)}$ the distance is defined as follows:
if $u=v$, then  $d_Z(z_i, z_j)=\frac{1}{\beta\gamma} \cdot d_T(z_i, z_j)$, otherwise (if $u\neq v$), $d_Z(z_i, z_j)=d_S(u, v) $.
\end{definition}
It is easily checked that the choice of the factor $1/(\beta\gamma)$ guarantees that $d_Z$ is indeed a metric.
For a finite metric space $S$ and an integer $t \geq 1$, let $[{S}]_{\beta}^t$ denote the metric space obtained by $\beta$-composition of $S$ with itself $t$ times.
The following theorem asserts that when the number of trees is small our upper bound on the distortion of Ramsey tree covers is tight up to a logarithmic factor.  Although the example is not described as a planar and doubling metric space, in Section \ref{app:rec_cycle_graph} we show that it can be approximated with constant distortion by a shortest path metric on a series-parallel graph with constant doubling dimension.
\begin{theorem}\label{thm:strong_lower_bound}
For any $k\geq 1$ and  large enough $n$, there is an $n$-point doubling metric space $X$, such that any Ramsey tree cover of $X$ of size $k$, has distortion $\Omega\left ( n^{\frac{1}{k}} \right)$.
\end{theorem}
\begin{proof}
Let $C_N$ denote the shortest path metric on the unweighted $N$-point cycle graph. For any integers $k,N\geq 1$, consider the metric space $Z_k(N)=\compKbeta{C_{N}}$, for any $\beta \geq 1/2$. We prove by induction on $k$, that any Ramsey tree cover of $Z_k(N)$ with $k$ trees has distortion at least $\frac{1}{3}|Z_k(N)|^{1/k}-1$.

For $k=1$ we have $Z_1(N)=C_N$, therefore by \cite{RR98} any embedding of $C_N$ into a tree metric requires distortion at least $N/3-1=\frac{1}{3}|Z_1(N)|-1$.  Assume the statement is true for $k-1$.

Let $T_1, T_2, \ldots, T_k$ be a Ramsey tree cover for $Z_k(N)$ with distortion $\alpha \geq 1$. Namely, for each point $v \in Z_k(N)$ there is at least one tree $T_i$ that preserves the distances of all the pairs containing $v$ up to a factor of $\alpha$. Such tree of minimal index is called a {\it home tree} of $v$.  Note that if each point of some subset $U \subseteq Z_k(N)$ has the same $T_i$ as its home tree, then the entire subspace $U$ embeds into $T_i$ with distortion $\alpha$.  For each  point $v \in Z_k(N)$ we associate a color $i \in [1,k]$: the color of $v$ is the index of its home tree. Then, all the points in $Z_k(N)$ are assigned a unique color in a way that partitions the metric into sub-metrics, where a submetric of color $i$ is embedded into a tree $T_i$ with distortion $\alpha$. We show that any coloring of the points of $Z_k(N)$ into $k$ colors has a large subspace of $Z_k(N)$ of the same color.
Note that $Z_k(N)$ is the composition of $C_{N}$ with $\left[C_{N}\right]_{\beta}^{k-1}$. This composition results in a metric space which consists of $N$ copies of $Z_{k-1}(N)$ scaled-down by a factor $1/\beta\cdot 2/N$, with the distances between any two points from different copies being the distance of their representatives in $C_N$. Denote these $N$ copies by $\{Z^j\}$, $1\leq j \leq N$.
Any coloring of the points of $Z_k(N)$ by $k$ colors must satisfy one of the following cases.
\begin{inparaenum}
\item There exists $1\le j\le N$ such that $Z^j$ is colored by at most $k-1$ colors;
\item Every $Z^j$ is colored by exactly $k$ colors.
\end{inparaenum}
In the first case, we get that there is an embedding of $Z^j$, which is a scaled-down copy of $Z_{k-1}(N)$, into a tree metric with distortion $\alpha$. By the induction hypothesis we obtain that $\alpha \geq \frac{1}{3}|Z_{k-1}(N)|^{1/{(k-1)}}-1=\frac{1}{3}N-1=\frac{1}{3}|Z_k(N)|^{1/k}-1$. In the second case, each metric $Z^j$ contains points with all $k$ colors, then by picking a point of color $1$, say, from each $Z^j$, we obtain a cycle of length $N$ embedded into the tree $T_1$ with distortion $\alpha$. Thus, by \cite{RR98} it must be that $\alpha \geq \frac{1}{3}N-1=\frac{1}{3}|Z_k(N)|^{1/k}-1$. Taking $n=N^k$ completes the proof.
\end{proof}
In the following section we show that $[C_{2N}]_\beta^k$ is, up to a constant distortion, a shortest path metric defined on a series-parallel graph $G=(V,E)$ with $O(1)$ doubling dimension. Thus, this lower bound holds for planar and doubling metrics as well.

\subsubsection{Recursive $N$-Cycle Graph}\label{app:rec_cycle_graph}
The recursive $N$-cycle graph is defined in a similar manner to the Laakso graph. $G_0$ is a single edge, $G_1$ is a cycle of length $2N$, with two additional vertices $s,t$ connected to two antipodals of the cycle with edges of weight $N$. We call the two additional edges {\em side edges}. For $i>0$, $G_i$ is constructed as follows: Take $G_1$, replace every cycle edge by a copy of $G_{i-1}$, and multiply the weight of the two side edges by $(3N)^{i-1}$. 
This is very similar to the Laakso graph \cite{L01}, but we do not recurse on the side edges, see \figureref{fig:cycle-graph}.

We record some properties: The graph $G_k$ is a series-parallel graph (and thus planar), and has doubling constant 6, being a finite subset of the Laakso graph \cite{L01}. It has 
diameter $(3N)^k$.

We view $G_k$ as $2N$ copies of $G_{k-1}$, denoted $Z^1,\dots,Z^{2N}$, arranged in a cycle, and the two side edges connected to $s,t$. Fix $1\le i< j\le 2N$, and let $C=\min\{j-i,2N+i-j\}$ be the distance in $C_{2N}$ between points $i,j$. Then the distance between any inner point (i.e., not the $s,t$ of that copy) in $Z^i$ to any inner point in $Z^j$ is at least $2C\cdot N\cdot (3N)^{k-2}$ and at most $2C\cdot N\cdot (3N)^{k-2} + (3N)^{k-1}=2C\cdot N\cdot (3N)^{k-2}(1+3/(2C))$. The former because there are $2C$ side edges of the copies of $G_{k-1}$ separating these copies, each of weight $N\cdot (3N)^{k-2}$, and the latter since the diameter of $G_{k-1}$ is $(3N)^{k-1}$. We see that the metric induced by taking 1 inner point from each $Z^i$ is up to a constant factor the metric of $C_{2N}$ scaled by a factor of $(3N)^{k-1}$.

Let $\gamma_k\cdot G_k$ denote the graph defined by $G_k$ multiplied by $\gamma_k:=1/(3N)^{k-1}$ (and $s,t$ are not in the image).

\begin{claim}\label{clm:comp:doubl}
For every $k\geq 1$, $[C_{2N}]_3^k$ embeds with constant distortion to the inner vertices of $\gamma_k\cdot G_k$.
\end{claim}
\begin{proof}
The proof is by induction. The base case $k=1$ holds trivially, since $G_1$ contains a $2N$ cycle and $\gamma_1=1$. By the induction hypothesis, each copy of $[C_{2N}]_3^{k-1}$ embeds with constant distortion to the inner points of $\gamma_{k-1}\cdot G_{k-1}$. Recall that $[C_{2N}]_3^k$ consists of $2N$ copies of $[C_{2N}]^{k-1}$ scaled by a factor of $1/(3N)$.
We can thus embed the $i$-th copy of $1/(3N)\cdot [C_{2N}]^{k-1}$ to the inner points of $1/(3N)\cdot\gamma_{k-1}\cdot Z^i=\gamma_k\cdot Z^i$ with constant distortion. Finally, we need to argue about distances between vertices of different copies. As mentioned in the previous paragraph, the distance between any inner point of $Z^i$ to any inner point in $Z^j$ is $\Theta(C\cdot(3N)^{k-1})$, (recall $C=\min\{j-i,2N+i-j\}$ is the distance in $C_{2N}$ between points $i,j$), so after scaling by $\gamma_k$ it becomes $\Theta(C)$, as required.
\end{proof}
We conclude that the lower bound of \theoremref{thm:strong_lower_bound} holds also for $G_k$, a series-parallel graph with $O(1)$ doubling dimension. Namely, any Ramsey tree cover of size $k$ requires distortion $\Omega(n^{1/k})$.
\begin{figure}
	\begin{center}
		\includegraphics[width=1.08\textwidth]{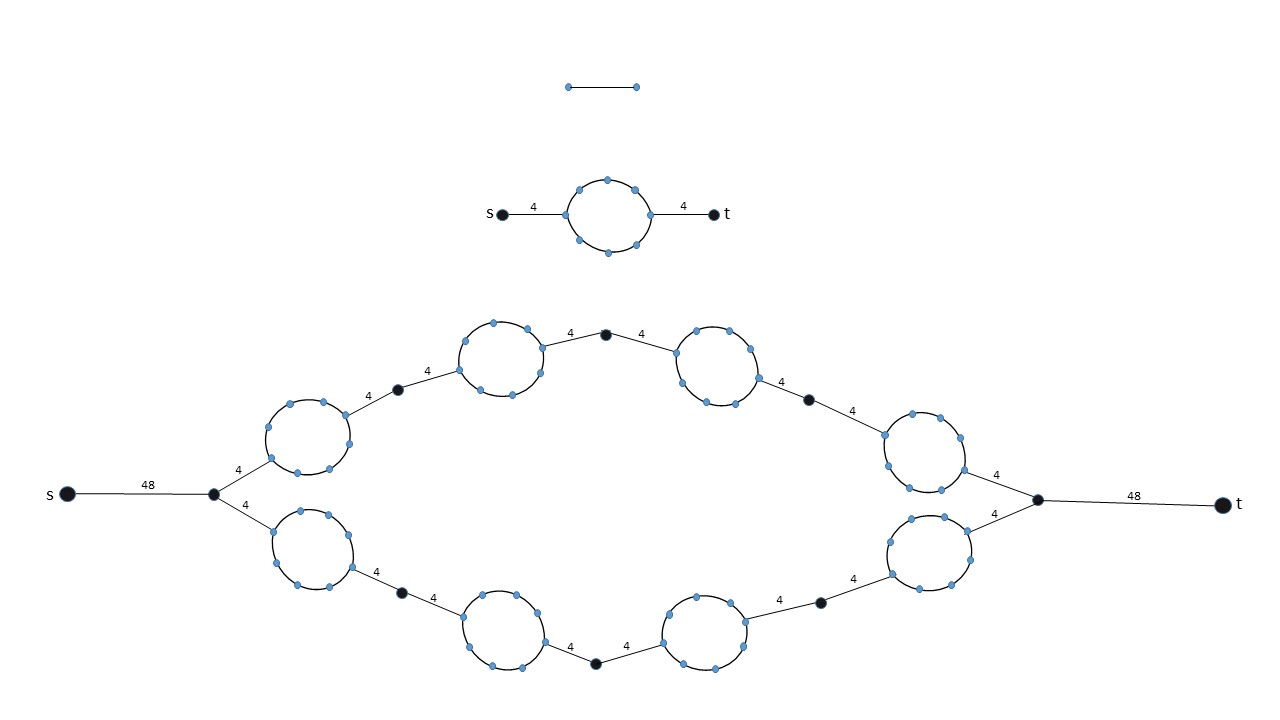}
		\caption{The first 3 levels of the recursive $4$-cycle graph.}
		\label{fig:cycle-graph}
	\end{center}
\end{figure}

\section{Open Problems}\label{sec:open}
Our work leaves several interesting open problems.
\begin{enumerate}
\item Tree cover for general metrics: In the regime where the number of trees $>>\log n$, we know a (nearly) tight tradeoff of distortion $\Theta(\alpha)$ with $O(\alpha\cdot n^{1/\alpha})$ trees. However, when the number of trees $k$ is small, the only lower bound on the distortion we are aware of is $\Omega(\log_kn)$. Is there a better lower bound? Or perhaps there exists a tree cover with $O(1)$ trees and distortion $O(\log n)$? What can one say about a tree cover with only 2 trees?

\item Tree covers for planar (and minor free) graphs: We show that using $O(1)$ trees one can obtain distortion $O(1)$. Can this be improved to distortion $1+\epsilon$? (Recall that we show here a cover of size $O((\log n)/\epsilon)^2$ for such distortion.)
\end{enumerate}

\section{Acknowledgements}
We are grateful to Michael Elkin and Shay Solomon for fruitful discussions.

\bibliographystyle{alpha}
\bibliography{fewTrees}

\appendix
\section{Proof of Lemma \ref{lem:landmarks}}\label{app:portal-proof}
We first describe how to choose the set of landmarks for each vertex.
Fix $x\in V$, and let $z_0\in P$ be the closest vertex of $P$ to $x$. Add $z_0$ to $L_x$. We traverse $P$ in one direction starting from $z_0$, and keep adding vertices to $L_x$. For integers $i>0$, assume we already added $\{z_0,\dots,z_{i-1}\}$, and let $z_i\in P$ be the first vertex we encounter on the path $P$ after $z_{i-1}$, that satisfies
\begin{equation}\label{eq:1}
d_G(x,z_i)<\frac{d_G(x,z_{i-1})+d_P(z_i,z_{i-1})}{1+\epsilon}~.
\end{equation}
Then we add $z_i$ to $L_x$ and continue with $i+1$.
We also do the symmetric thing for $i<0$ while traveling in the opposite direction starting from $z_0$. Next we show that we can only add a small number of vertices to $L_x$. For $i>0$:
\begin{eqnarray*}
d_G(x,z_i)&<&\frac{d_G(x,z_{i-1})+d_P(z_i,z_{i-1})}{1+\epsilon}\\
&\le&(1-\epsilon/2)\cdot d_G(x,z_{i-1})+d_P(z_i,z_{i-1})\\
&\le&d_G(x,z_{i-1})-\epsilon/2\cdot d_G(x,z_0)+d_P(z_i,z_{i-1})~,
\end{eqnarray*}
where the first inequality uses that $\epsilon<1$ and the second that $z_0$ is the closest vertex in $P$ to $x$. Applying this calculation iteratively, and using the fact that $P$ is a shortest path, we conclude that
\[
d_G(x,z_i)<d_G(x,z_0)-i\cdot\epsilon/2\cdot d_G(x,z_0)+d_P(z_i,z_0)~.
\]
Note that if $i>4/\epsilon$ we get that
\[
d_G(x,z_i)<d_G(x,z_0)-2 d_G(x,z_0)+d_P(z_i,z_0)=d_P(z_i,z_0)-d_G(x,z_0)~,
\]
which contradicts the fact that $P$ is a shortest path (the distance from $z_0$ to $z_i$ is shorter going through $x$). The same argument bounds the number of landmarks added from the other direction, so their total number is at most $8/\epsilon$.

It remains to check that if the shortest path $\pi$ from $x$ to $y$ intersects $P$, then there exist $u\in L_x$ and $v\in L_y$ satisfying
\[
d_G(x,u)+d_P(u,v)+d_G(v,y)\le(1+\epsilon)d_G(x,y)~.
\]
Let $u'$ (respectively, $v'$) be the first (resp., last) vertex on $\pi$ that lies on $P$. Consider first $u'$, we want to show there is a $u\in L_x$ so that $d_G(x,u)+d_P(u,u')\le (1+\epsilon)\cdot d_G(x,u')$. If $u'\in L_x$ we are done, otherwise it lies after some landmark $u:=z_i\in L_x$. Since $u'$ was not chosen to $L_x$, we have by \eqref{eq:1} that
\[
(1+\epsilon)\cdot d_G(x,u')\ge d_G(x,u)+d_P(u,u')~.
\]
The same argument asserts there is $v\in L_y$ such that
\[
(1+\epsilon)\cdot d_G(y,v')\ge d_G(y,v)+d_P(v,v')~.
\]
Combining these bounds we get that
\begin{eqnarray*}
\lefteqn{d_G(x,u)+d_P(u,v)+d_G(v,y)}\\&\le& ((1+\epsilon)\cdot d_G(x,u')-d_P(u,u'))+d_P(u,v)+((1+\epsilon)\cdot d_G(y,v')-d_P(v,v'))\\
&\le& (1+\epsilon)\cdot d_G(x,u') +(1+\epsilon)\cdot d_G(y,v') +d_P(u',v')\\
&\le&(1+\epsilon)\cdot d_G(x,y)~.
\end{eqnarray*}

\end{document}